\documentclass{revtex4}
\usepackage{amsmath,amsthm,amssymb}
\usepackage{rotating}
\usepackage{multirow}

\usepackage{graphicx}
\usepackage{amsfonts}
\newtheorem{theorem}{Theorem}[section]

\newtheorem{proposition}[theorem]{Proposition}
\newtheorem{cor}[theorem]{Corollary}
\theoremstyle{remark}
\newtheorem{remark}[theorem]{Remark}

\theoremstyle{definition}
\newtheorem{definition}[theorem]{Definition}

\theoremstyle{example}
\newtheorem{example}[theorem]{Example}

\theoremstyle{notation}

\newcommand{\ket}[1]{|#1\rangle}

\begin{document}

\title{Exterior calculus and fermionic quantum computation}            
\author{A. Vourdas\\
Department of Computer Science\\University of Bradford\\ Bradford BD7 1DP, UK}

\begin{abstract}
Exterior calculus with its three operations meet, join  and hodge star complement, is used for the representation of fermion-hole systems and for fermionic analogues of 
logical gates. 
Two different schemes that implement fermionic quantum computation, are proposed.
The first scheme compares fermionic gates with Boolean gates, and leads to novel electronic devices that simulate fermionic gates. The second scheme uses
a well known map between fermionic and multi-qubit systems, to simulate fermionic gates within multi-qubit systems.
\end{abstract}

\maketitle

\section{ Introduction}

Exterior calculus  is based on Grassmann's progressive and regressive products\cite{G}.
Much of the literature uses exterior calculus with only one of its operations, the exterior product.
Rota and collaborators \cite{R0,R1,R2,R3} studied in detail exterior calculus 
with all its three operations in a modern pure mathematics context, and its applications to projective geometry.
This formalism is used in a physical context in this paper for a description of fermion-hole systems.

Exterior calculus is defined on a $d$-dimensional vector space $V(d)$ over the field $\mathbb C$ of complex numbers, endowed with 
sums of oriented volumes in $V(d)$ and its subspaces. 
In a physical context the oriented volumes are interpreted as Slater determinants, which describe a system of $d$ fermions.

There are three operations:
\begin{itemize}
\item
The meet operation describes how to combine two fermionic systems, into a larger fermionic system.
\item
The join operation joins systems of holes. 
\item
The Hodge star complement maps a system of fermions,
into the corresponding system of holes, and it changes
the meet (join) operation of two fermionic subsystems, into the join (meet) operation between the corresponding subsystems of holes.
\end{itemize}
These three operations are the fermionic analogues of the logical operations AND, OR, NOT in lattice theory.

Exterior calculus is not a lattice, it is a formalism based on determinants and both the concept of independence and also antisymmetry are deeply embedded in it.
The formalism expresses fermionic logic and provides the theoretical foundation for fermionic quantum computation\cite{K,K1}, in analogy to Boolean algebra 
which is the   theoretical foundation for classical computation, and to the Birkhoff-von Neumann lattice \cite{LO1,LO2,LO3,LO4,LO5,LO6} which is the  
 theoretical foundation for general quantum computation.

A given computational task might have different complexity within classical, bosonic or fermionic computation\cite{LL,LL1,LL2}.
Efficient computations require gates that perform tasks close to  the overall task.
Otherwise it is a complex and expensive process to reduce the overall task into many small tasks that can be performed by the gates.  
An algebraic structure that describes the logic of fermion-hole systems, and its experimental implementation,  is highly desirable
for computations that involve fermions, 
For example, it might be helpful to current work on fermionic simulations for quantum chemistry (e.g., \cite{B,K2})
The fermion-hole formalism can also be useful in other areas, like condensed matter, physical electronics for semiconductors, etc.

In this paper we use in a physical context the full exterior calculus with its three operations.
Our emphasis is in the physical interpretation of the formalism.
Theorems which are rigorously proved in the mathematics literature\cite{R0,R1,R2,R3}  are given here without proof, but they are interpreted physically in the context of fermionic systems.
We then suggest two different schemes for the implementation of fermionic quantum computation, as follows:
\begin{itemize}
\item
A comparison of fermionic gates with Boolean gates is made. This leads to novel electronic devices,
which are variants of the Boolean gates (AND, OR), and can be used for the simulation of fermionic logic. 
\item
There is a well known map between fermionic and multi-qubit systems.
We use this to introduce a fermionic (exterior calculus) logic in multi-qubit systems.
Fermionic gates might speed up computations for fermionic systems.
Implementation of fermionic gates with genuine fermionic systems might be a longer term task experimentally.
A fermionic logic formalism within multi-qubit systems (as discussed in section \ref{fermi}), might be easier to implement with existing technologies.
\end{itemize}

In section II we introduce the join operation and interpret it as the joining of systems of fermions.
We explain how Slater determinants, antisymmetry and the Pauli exclusion principle are embedded in the formalism.
In section III we introduce the Hodge star complement operation, and  interpret it as a map where fermions are replaced by holes, and holes by fermions.  
In section IV we introduce the meet operation and interpret it as the joining of systems of holes.
In section V we show how the scalar product is introduced within exterior calculus.

In section VI we compare and contrast fermionic gates with Boolean gates.
This leads to novel electronic devices,
which can be used for the implementation of fermionic logic. 
In section VII we show how we can simulate fermionic (exterior calculus) logic within multi-qubit systems.
We conclude in section VIII with a discussion of our results.

\section{The meet operation: joining states of fermions}

\subsection{Notation}
Rota and collaborators\cite{R0,R1,R2,R3} use the join notation $\vee$ for exterior products in oriented volumes, and the meet notation $\wedge$ for its dual operation.
Most of the literature uses the opposite notation, and this is what we adopt here: $\wedge$ for exterior product and $\vee$ for its dual operation defined later in section \ref{meet}.
We will interpret physically, the operation $\wedge $ as joining systems of fermions, and the operation $\vee$ as joining systems of holes.

Most of the literature uses only one operation, the exterior product (e.g., in geometry \cite{geo}).
Bourbaki \cite{BOU} uses mainly the exterior product (denoted as $\wedge$), and mentions briefly the dual operation in an exercise.
In this paper both of these operations play an equally important role, one in relation with fermions and the other in relation with holes.
This ensures equal treatment of both fermions and holes.
We also use the Hodge star complement, and show that these three operations together, describe the logic of fermionic systems. 

Holes are absent fermions in some of the available sites. We can interpret everything in terms of fermions or in terms of holes.
The analogue of this in Boolean algebra is that we can work with sets or equivalently with their complements.
For example, we can register the students who attend a lecture, or equivalently we can register the students who are absent from a lecture.

\subsection{Slater determinant for a system of $d$ fermions}

Let $x_1,...,x_d$ be vectors in $V(d)$ in a certain order. We use the notation $x_1\wedge ...\wedge x_d$ for the oriented volume associated with these vectors. In Grassmann's terminology this is the progressive product.
Cartan called them multivectors and discussed their properties\cite{cartan}.
In a physical context we will link them to (and refer to them as) Slater determinants.

If $e_1,...,e_d$ is an orthonormal basis in $V(d)$, then 
\begin{eqnarray}
X=x_1\wedge ...\wedge x_d=\det(x_1,...,x_d)e_1\wedge ...\wedge e_d.
\end{eqnarray}
Here $\det(x_1,...,x_d)$ is the determinant with columns the components of the vectors $x_1,...,x_d$ in the orthonormal basis $e_1,...,e_d$. 
More generally, if $A$ is a $d\times k$ matrix and $B$ is a $d\times (d-k)$ matrix, we will use the notation $\det (A,B)$ for the determinant of the combined $d\times d$ matrix $(A,B)$. 
The use of an orthonormal basis avoids the need for covariant and contravariant vectors, which do not add anything to the physics of the present paper.

The concept of linear independence is deeply embedded in exterior calculus.
$X$ is non-zero only if the vectors $x_1,...,x_d$ are linearly independent. 
We interpret $X$ as the Slater determinant of $d$ fermions located in the sites $r_1,...,r_d$,
which are taken to be fixed (e.g., fermionic lattice).
The vector $x_i$ is given by
\begin{eqnarray}
x_i=
\begin{pmatrix}
\psi_i(r_1)\\
\vdots\\
\psi_i(r_d)
\end{pmatrix}
\end{eqnarray} 
where $\psi _i(r_j)$ is the wavefunction of the $i$-fermion in the position $r_j$.
The meet operation makes Slater determinants for  a system of $d$ fermions (up to a normalization factor).

Let $\sigma$ be a permutation of the indices $1,...,d$ into $\sigma _1,...,\sigma_d$, with sign ${\rm sgn} (\sigma)$ (which is $-1$ for odd permutations, and $+1$ for even permutations). Then
\begin{eqnarray}
X=x_1\wedge ...\wedge x_d={\rm sgn} (\sigma) x_{\sigma _1}\wedge ...\wedge x_{\sigma _n}
\end{eqnarray} 
This and many other relations, are based on the usual properties of determinants.
An example is the multilinearity property:
\begin{eqnarray}
x_1\wedge ...\wedge (ax_i+by_i)\wedge...\wedge x_d=a (x_1\wedge ...\wedge x_i\wedge...\wedge x_d)+b (x_1\wedge ...\wedge y_i\wedge...\wedge x_d).
\end{eqnarray}

Below we will use the notation
\begin{eqnarray}
{\cal E}=e_1\wedge ...\wedge e_d.
\end{eqnarray}
${\cal E}$ is the hole vacuum, or equivalently all states in it are occupied by fermions.
$1$ is the fermionic vacuum, or equivalently all states in it are occupied by holes.

\subsection{Subsystems of fermions: Slater determinants in subspaces of $V(d)$ }
We consider a $k$-dimensional subspace $V(k)$ of $V(d)$, and an orthonormal basis $e_1,...,e_d$ such that the $e_1,...,e_k$ span the subspace $V(k)$.
Let $y_1,...,y_k$ be vectors in $V(k)$. The Slater determinant associated with these vectors is
\begin{eqnarray}\label{11}
Y=y_1\wedge ...\wedge y_k=\det(y_1,...,y_k)e_1\wedge ...\wedge e_k.
\end{eqnarray}
Here $\det(y_1,...,y_k)$ is the determinant with columns the $k$ components of the vectors $y_1,...,y_k$ in the orthonormal basis $e_1,...,e_k$.
We can also regard the $y_1,...,y_k$ as $d$-dimensional vectors in $V(d)$ with the components labelled with $k+1,...,d$ equal to zero.
Then Eq.(\ref{11}) can be written as
\begin{eqnarray}
Y=y_1\wedge ...\wedge y_k=\det(y_1,...,y_k,e_{k+1},...,e_d)e_1\wedge ...\wedge e_k.
\end{eqnarray}
We call Grassmann space of step $k$ and denote as $\bigwedge {^k}[V(d)]$ , the space generated by these Slater determinants. 
It describes subsystems of $k$ fermions, within a bigger system of $d$ fermions. Its elements are called tensors of step $k$.
Its dimension is 
$\begin{pmatrix} 
d\\
k\\
\end{pmatrix}$.

A tensor $x$ of step $k$ is called decomposable or extensor if there exist $k$ vectors  $x_1,...,x_k$ such that
$X=x_1\wedge ...\wedge x_k$ (it is a single Slater determinant). 
Cartan\cite{cartan} calls them `simple multivectors'.
Not all elements of $\bigwedge ^k[V(d)]$ are extensors (they are sums of many extensors).
The general element in $\bigwedge {^k}[V(d)]$ is
\begin{eqnarray}
X=\sum x(i_1,...,i_k)e_{i_1}\wedge ...\wedge e_{i_k};\;\;\;x(i_1,...,i_k)\in {\mathbb C}.
\end{eqnarray}
Here $\{i_1,...,i_k\}$ is any subset with cardinality $k$, of the set of indices $\{1,...,d\}$.
The vectors $e_i$ obey the relations
\begin{eqnarray}
e_i\wedge e_i=0;\;\;\;e_i\wedge e_j+e_j\wedge e_i=0.
\end{eqnarray}

\begin{example} 
$\bigwedge ^2[V(4)]$ describes subsystems of two fermions, within the four-dimensional space $V(4)$ which can accommodate a maximum of four fermions.
The $e_1\wedge e_2+e_3\wedge e_4$ is the sum of two Slater determinants describing the subsystem of two fermions in the states $1,2$, and the subsystem of two fermions in the states $3,4$.
\end{example}
We note here that ref.\cite{S} has introduced the concept of Slater rank in connection with entanglement \cite{S,S1} in fermionic systems. This counts the number of  Slater determinants in a sum.

The extensor $X=x_1\wedge ...\wedge x_k$ is associated with a $k$-dimensional subspace of $V(d)$ spanned by the vectors
$x_1,...,x_k$ which we denote as ${V}(k;X)$.
This notation indicates which $k$-dimensional subspace of $V(d)$ we consider.
\begin{example}
In $\bigwedge {^2}[V(3)]$ we consider the extensor
\begin{eqnarray}\label{79}
A=(\alpha e_1+\beta e_2)\wedge (\gamma e_2+\delta e_3)=
\alpha \gamma e_1\wedge e_2+\alpha \delta e_1\wedge e_3+\beta \delta e_2\wedge e_3.
\end{eqnarray}
It can be regarded as an oriented area in the plane $V(2;X)$ defined by the vectors
$(\alpha e_1+\beta e_2, \gamma e_2+\delta e_3)$.
Physically it is the superposition of the two fermion states $e_1\wedge e_2, e_1\wedge e_3, e_2\wedge e_3$ with the coefficients 
$\alpha \gamma, \alpha \delta, \beta \delta$, correspondingly.
\end{example}

\paragraph*{Independence, Pauli exlusion principle, and the physical meaning of zero result:} 
When the meet operation of several vectors is $0$, mathematically this means that the vectors are not independent.
Physically it means that the Slater determinant is zero, which reflects the fact that the Pauli exclusion principle is violated.
This is also true for the other operations introduced later.
Zero result means that this operation is physically impossible.
We stress that zero result does not mean that we get the fermionic vacuum.
The fermionic  vacuum is represented with a number in ${\mathbb C}\setminus \{0\}$, which can be normalized to $1$.

\subsection{The Grassmann space $\bigwedge[V(d)]$}

We might have a more general case, of sums of Slater determinants which belong to various Grassmann spaces with different steps $k$ (where $k\le d$).
In this case the number of fermions is not fixed.

We call Grassmann space the
\begin{eqnarray}\label{1}
\bigwedge [V(d)]=\bigoplus _{k=0}^d\bigwedge {^k}[V(d)];\;\;\;
\bigwedge {^0}[V(d)]={\mathbb C};\;\;\;\bigwedge {^1}[V(d)]=V(d).
\end{eqnarray}  
Its dimension is 
\begin{eqnarray}
\sum _{k=0}^d
\begin{pmatrix} 
d\\
k\\
\end{pmatrix}
=2^d.
\end{eqnarray}
The general element in $\bigwedge [V(d)]$ is
\begin{eqnarray}
A=\sum a(i_1,...,i_k)e_{i_1}\wedge ...\wedge e_{i_k};\;\;\;a(i_1,...,i_k)\in {\mathbb C},
\end{eqnarray}
where $\{i_1,...,i_k\}$ is any subset of the set of indices $\{1,...,d\}$.\paragraph*{Fock and Grassmann spaces:}
Let $F_k$ be the Fock space describing $k$ fermions. 
In particular $F_0$ contains the vacuum $\ket{0}$.
The full Fock space is  
\begin{eqnarray}\label{2V}
F=F_0\oplus F_1\oplus F_2 \oplus....
\end{eqnarray}
The Grassmann space $\bigwedge {^k}[V(d)]$ of step $k$, is analogous to $F_k$.
But $F_k$ is infinite dimensional, in contrast to 
$\bigwedge {^k}[V(d)]$ which is finite-dimensional.
Similarly the Grassmann space $\bigwedge[V(d)]$ is analogous to the full Fock space $F$. 
But the direct sum in Eq.(\ref{1}) is finite and the $\bigwedge[V(d)]$ is $2^d$-dimensional.
In contrast  the direct sum in Eq.(\ref{2V}) is infinite, and $F$ is infinite-dimensional.

\subsection{Superselection rules}
In some problems there are superselection rules which restrict the system into a subspace of the full Grassmann space.
An example is that the number of fermions modulo $2$ should have a definite value (either $0$ or $1$).
In this case the Grassmann space can be written as 
\begin{eqnarray}
\bigwedge [V(d)]=
\bigwedge ^{\rm even}[V(d)]\oplus
\bigwedge ^{\rm odd}[V(d)];\;\;\;
\bigwedge ^{\rm even}[V(d)]=\bigoplus _{k=0}^a\bigwedge {^{2k}}[V(d)];\;\;\;
\bigwedge ^{\rm odd}[V(d)]=\bigoplus _{k=0}^b\bigwedge {^{2k+1}}[V(d)].
\end{eqnarray}  
Here for odd $d$ we have $a=b=(d-1)/2$, and for even $d$ we have $a=d/2$ and $b=(d-2)/2$. 
The state of the system belongs to either the `odd Grassmann space' or to the `even Grassmann space'.
Superpositions of states in different subspaces are not allowed.

Other superposition rules (e.g., related to electric charge) might be relevant in some problems.
Below we consider the full Grassmann space, but we stress that in the presence of superselection rules, we need to work in the appropriate subspace.

\subsection{Properties of the meet operation}

Physically, $A\wedge B$ combines two fermionic systems, into a larger fermionic system.
The following two propositions are given without mathematical proof\cite{R0,R1,R2,R3}, but we interpret them physically.
\begin{proposition}\label{pro1}
Let ${\cal O}$ be the vector space that contains only the zero vector, and $A,B$ be extensors associated with the subspaces $V(k,A), V(\ell, B)$.
\begin{itemize}
\item
If  $V(k,A)\cap V(\ell, B)\ne {\cal O}$, then $A\wedge B=0$.

\item
If  $V(k,A)\cap V(\ell, B)= {\cal O}$, then $A\wedge B\ne 0$ and  the subspace associated with the extensor $A\wedge B$ is ${\rm span}[V(k,A)\cup V(\ell, B)]$.
\end{itemize}
\end{proposition}

\paragraph*{Physical interpretation:}
The above proposition states that we can join two fermionic subsystems into a larger fermionic system, only if they `live' in different 
non-overlapping subspaces of $V(d)$. Because in this case the two subsystems occupy different states, and respect the Pauli exclusion principle.

\begin{cor}
If $A$ is an extensor with ${\rm step}(A) \ge 1$ then  $A\wedge A=0$. Also $1\wedge 1=1$.
\end{cor}
\paragraph*{Physical interpretation:}
We cannot join a fermionic system with itself because this violates the Pauli exclusion principle.
In $1\wedge 1=1$ we join the fermionic vacuum with itself.

\begin{proposition}
\begin{itemize}
\item[(1)]
The associativity property holds:
\begin{eqnarray}
A\wedge(B\wedge C)=(A\wedge B)\wedge C.
\end{eqnarray}
\item[(2)]
If $A\in \bigwedge {^k}[V(d)]$ and $B\in \bigwedge {^\ell}[V(d)]$ then 
\begin{eqnarray}\label{300}
A\wedge B=(-1)^{k\ell}B\wedge A.
\end{eqnarray}
\end{itemize}
\end{proposition}
\paragraph*{Physical interpretation:}Eq.(\ref{300}) expresses the antisymmetry property of fermions.

The following remark is useful for practical calculations of the join operation.
\begin{remark} 
If $e_1,...,e_d$ is an orthonormal  basis in $V(d)$ we consider the extensors
\begin{eqnarray}
A=e_{i_1}\wedge...\wedge e_{i_r};\;\;\;B=e_{j_1}\wedge...\wedge e_{j_s}.
\end{eqnarray}
where the $\{i_1,...,i_r\}$ and $\{j_1,...,j_s\}$ are subsets of the set of all indices $\{1,...,d\}$.
Then
\begin{itemize}
\item
if $\{i_1,...,i_r\}\cap \{j_1,...,j_s\}\ne \emptyset$ then $A\wedge B=0$ 
\item
if $\{i_1,...,i_r\}\cap \{j_1,...,j_s\}=\emptyset$ then $A\wedge B$ is given by 
\begin{eqnarray}
A\wedge B=\lambda \bigwedge _{k\in S}e_k;\;\;\;S=\{i_1,...,i_r\} \cup \{j_1,...,j_s\}
\end{eqnarray}
where $\lambda=\pm1$ depending on the order of the indices.
\end{itemize}
In particular for every extensor $A$ with ${\rm step}(A)\ge 1$:
\begin{eqnarray}
A\wedge {\cal E}=0.
\end{eqnarray}
Physically the fermions in ${\cal E}$ fill all the states in the $d$-dimensional space $V(d)$,
and there is no room compatible with the Pauli exclusion principle, for the extra fermions in $A$.
So the operation $A\wedge {\cal E}$ is physically impossible, and this the meaning of the zero result.
 If $A=1$ (the vacuum state), then
\begin{eqnarray}
 1\wedge {\cal E}={\cal E}.
\end{eqnarray}
Also for every extensor $A$ 
\begin{eqnarray}
A\wedge 1=A.
\end{eqnarray}
Physically, $1$ represents the fermionic vacuum  and if we join it with the system described by $A$, we get $A$.
These relations are summarized in table \ref{t1}.
\end{remark}

\begin{table}
\caption{Some relations for extensors $A,B$ in $\bigwedge [V(d)]$. The linearity property can be used to find analogous relations for more general tensors. $1$ is the fermionic vacuum and ${\cal E}$ is the vacuum for holes. If an operation gives the result $0$, it means that it is physically impossible because of the Pauli exclusion principle. }
\def\arraystretch{2}
\begin{tabular}{|c|c|}\hline
$A\wedge(B\wedge C)=(A\wedge B)\wedge C$&$A\vee(B\vee C)=(A\vee B)\vee C$\\\hline
$e_1\wedge ...\wedge e_d={\cal E}$&$e_1\vee ...\vee e_d=1$\\\hline
$A\in \bigwedge {^k}[V(d)]$;\;\;$B\in \bigwedge {^\ell}[V(d)]$&$A\in \bigwedge {^k}[V(d)]$;\;\;$B\in \bigwedge {^\ell}[V(d)]$\\
$A\wedge B=(-1)^{k\ell}B\wedge A$&$A\vee B=(-1)^{(d-k)(d-\ell)} B\vee A$\\\hline
$\star (A\wedge B)=(\star A)\vee (\star B)$&$\star(A\vee B)=(\star A) \wedge (\star B)$\\\hline
${\rm step}( A) \ge 1\;\rightarrow\; A\wedge {\cal E}=0$ &$A\vee {\cal E}=A$\\\hline
$A\wedge 1=A$&${\rm step}(A) \le d-1\;\rightarrow\;A\vee 1=0$\\\hline
$1\wedge {\cal E}={\cal E}$&$1\vee {\cal E}=1$\\\hline
${\rm step}( A) \ge 1\;\rightarrow\; A\wedge A=0$ &${\rm step}(A) \le d-1\;\rightarrow\;A\vee A=0$\\\hline
$1\wedge 1=1$&$1\vee 1=0$\\\hline
${\cal E}\wedge {\cal E}=0$&${\cal E}\vee {\cal E}={\cal E}$\\\hline
$A\in \bigwedge {^k}[V(d)]$;\;\;\;$B\in \bigwedge {^{d-k}}[V(d)]$&$A\in \bigwedge {^k}[V(d)]$;\;\;\;$B\in \bigwedge {^{d-k}}[V(d)]$\\
$A\wedge B=(A\vee B){\cal E}=\det(A,B){\cal E}$&$A\vee B=\det(A,B)$\\\hline
$A\wedge (\star A)=\det(A,\star A){\cal E}$&$A\vee (\star A)=\det (A,\star A)$\\\hline
\end{tabular} \label{t1}
\end{table}

\begin{example}
We consider $\bigwedge [V(2)]$ and the basis $1, e_1, e_2, {\cal E}$. 
There are two orthogonal states in $V(2)$, and: 
\begin{itemize}
\item
$1$ represents the fermionic vacuum;
\item
in $e_1$ one fermion occupies the first state; 
\item
in  $e_2$ one fermion occupies the second state;
\item
in ${\cal E}$ one fermion occupies the first state and a second fermion  the second state.
\end{itemize}
In table \ref{tq} we present the $A\wedge B$ where $A,B$ are extensors in this basis.
 The result $0$ indicates that this operation is physically impossible.
Using this table and the multilinearity property, we can easily find the $A\wedge B$ for all tensors in  $\bigwedge [V(2)]$.
\end{example}
\begin{table}
\caption{$A\wedge B$ and $A\vee B$, for $A,B \in \bigwedge [V(2)]$.
The result $0$ indicates that this operation is physically impossible because of the Pauli exclusion principle.
$A\wedge B$ can be viewed as an OR gate for fermions, or equivalently as an AND gate for holes.
$A\vee B$ can be viewed as an AND gate for fermions, or equivalently as an OR gate for holes.
A comparison with Boolean gates is discussed in section \ref{gates} and table \ref{tt}.}
\def\arraystretch{2}
\begin{tabular}
{|c|c|c|c|}\hline
$\;\;A\;\;$&$\;\;B\;\;$&$A\wedge B$&$A\vee B$\\\hline
$1$&$1$&$1$&$0$\\\hline
$1$&$e_1$&$e_1$&$0$\\\hline
$1$&$e_2$&$e_2$&$0$\\\hline
$1$&${\cal E}$&${\cal E}$&$1$\\\hline
$e_1$&$1$&$e_1$&$0$\\\hline
$e_1$&$e_1$&$0$&$0$\\\hline
$e_1$&$e_2$&${\cal E}$&$1$\\\hline
$e_1$&${\cal E}$&$0$&$e_1$\\\hline
$e_2$&$1$&$e_2$&$0$\\\hline
$e_2$&$e_1$&$-{\cal E}$&$-1$\\\hline
$e_2$&$e_2$&$0$&$0$\\\hline
$e_2$&${\cal E}$&$0$&$e_2$\\\hline
${\cal E}$&$1$&${\cal E}$&$1$\\\hline
${\cal E}$&$e_1$&$0$&$e_1$\\\hline
${\cal E}$&$e_2$&$0$&$e_2$\\\hline
${\cal E}$&${\cal E}$&$0$&${\cal E}$\\\hline
\end{tabular}\label{tq}
\end{table}
\begin{example}
We consider the $3$-dimensional space vector $V(3)$, and an orthonormal  basis $e_1, e_2, e_3$. Then:
\begin{itemize}
\item
$\bigwedge {^0}[V(3)]$ is the one-dimensional space $\mathbb C$. It describes the fermionic vacuum.
\item
$\bigwedge {^1}[V(3)]$ is the $3$-dimensional space generated by $e_1$, $e_2$, $e_3$, and describes one fermion states.
\item
$\bigwedge {^2}[V(3)]$ is the $3$-dimensional space generated by $e_1\wedge e_2$, $e_1\wedge e_3$, $e_2\wedge e_3$.
It contains oriented areas in planes within $V(3)$, which describe two fermion (antisymmetric) states.
\item
$\bigwedge {^3}[V(3)]$ is the one-dimensional space generated by ${\cal E}$, and describes three fermion (antisymmetric) states.
\end{itemize}
The exterior calculus space $\bigwedge [V(3)]$ is $8$-dimensional, with general element
\begin{eqnarray}
A=a_0+a_1e_1+a_2e_2+a_3e_3+a_4 e_1\wedge e_2+a_5e_1\wedge e_3+a_6 e_2\wedge e_3 +a_7{\cal E};\;\;\;a_i\in \mathbb C.
\end{eqnarray}

We next calculate the $X\wedge e_1$ where $X$ is given in Eq.(\ref{79}), and we find
$X\wedge e_1=\beta \delta {\cal E}$.
\end{example}

\subsection{Creation and annihilation operators for fermions and holes}
We relate briefly the above formalism with creation and annihilation operators.
$a_k$ are annihilation operators of fermions, or equivalently creation operators of holes.
$a_k^{\dagger}$ are creation operators of fermions, or equivalently annihilation operators of holes.
They obey the anticommutation relations:
\begin{eqnarray}
\{a_j, a_k^{\dagger}\}=\delta _{jk}{\bf 1};\;\;\;\{a_j^{\dagger}, a_k^{\dagger}\}=\{a_j, a_k\}=0;\;\;\;j,k \in {\cal I}=\{1,...,d\}
\end{eqnarray}
We consider the subset $S=\{i_1,...,i_k\}$ of the set of indices ${\cal I}$, and its complement $S^C={\cal I}\setminus S$.
The indices are in ascending order.
We then consider the operator $[{\cal A}(i_1,...,i_k)]^{\dagger}$ (${\cal A}(i_1,...,i_k)$) which 
has creation (annihilation) operators in the indicated places.
For example if $d=4$, then $[{\cal A}(1,4)]^{\dagger}=a_1^{\dagger}a_4^{\dagger}$ and ${\cal A}(1,4)=a _1a _4$.

The operator $[{\cal A}(i_1,...,i_k)]^{\dagger}$ maps the fermionic vacuum $1$, into the following system of $k$ fermions:
\begin{eqnarray}
[{\cal A}(i_1,...,i_k)]^{\dagger}\;:\;1\;\rightarrow\;\bigwedge _{j\in S}e_j
\end{eqnarray}
The operator ${\cal A}(i_1,...,i_k)$ maps the hole vacuum ${\cal E}$, into the following system of $k$ holes:
\begin{eqnarray}
{\cal A}(i_1,...,i_k)\;:\;{\cal E}\;\rightarrow\;\bigwedge _{j\in S^C} e_j.
\end{eqnarray}

\section{The Hodge star complement: fermions are mapped into holes}

\begin{definition}
Let $S=\{i_1,...,i_k\}$ be a subset of the set of all indices ${\cal I}=\{1,...,d\}$, and $S^C=
\{j_1,...,j_{d-k}\}$
its complement. The elements are ordered so that $i_1<...<i_k$ and $j_1<...<j_{d-k}$.
The Hodge star complement is denoted with $\star$, and is a bijective map from $\bigwedge^{k}[V(d)]$ onto $\bigwedge ^{d-k}[V(d)]$ 
(they have the same dimension), as follows:
\begin{eqnarray}\label{com}
\star \left [e_{i_1}\wedge ...\wedge e_{i_k}\right ]=(-1)^{i_1+...+i_k-k(k+1)/2}e_{j_1}\wedge ...\wedge e_{j_{d-k}}
\end{eqnarray}
For more general tensors in $\bigwedge[V(d)]$ the Hodge star complement is defined through the multi-linearity property.
\end{definition}
This map replaces fermions by holes (absent fermions) and holes by fermions.
For example the $e_{i_1}\wedge ...\wedge e_{i_k}$ has fermions labelled by the set $S$, and has no fermions (i.e., it has holes) labelled by the set $S^C$.
Its Hodge star complement $e_{j_1}\wedge ...\wedge e_{j_{d-k}}$ has fermions labelled by the set $S^C$, and has no fermions (i.e., it has holes) labelled by the set $S$.

We note that
\begin{eqnarray}
\star 1={\cal E};\;\;\;\star {\cal E}=1.
\end{eqnarray}
 Also
\begin{eqnarray}
\star (\star A)=(-1)^{k(d-k)}A.
\end{eqnarray}

\begin{remark}
The Hodge star complement is analogous to negation in Boolean algebra, in the sense that the holes are absences (`negations') of fermions.
\end{remark}
\begin{example}
In $ \bigwedge [V(2)]$:
\begin{eqnarray}
&&\star 1=e_1\wedge e_2={\cal E};\;\;\;\star e_1=e_2;\;\;\;\star e_2=-e_1;\;\;\;\star {\cal E}=1. 
\end{eqnarray}
\end{example}\begin{example}
In $ \bigwedge [V(3)]$:
\begin{eqnarray}
&&\star 1=e_1\wedge e_2\wedge e_3={\cal E};\;\;\;\star e_1=e_2\wedge e_3;\;\;\;\star e_2=-e_1\wedge e_3;\;\;\;\star e_3=e_1\wedge e_2\nonumber\\
&&\star (e_1\wedge e_2)=e_3;\;\;\;\star (e_1\wedge e_3)=-e_2;\;\;\star (e_2\wedge e_3)=e_1;\;\;\;\star {\cal E}=1. 
\end{eqnarray}
\end{example}

\subsection{Covectors: systems with one hole}
All tensors in  $\bigwedge ^{1}[V(d)]$ and $\bigwedge ^{d-1}[V(d)]$ are decomposable and they are called 
vectors and covectors, correspondingly.
Physically, vectors represent systems with one fermion.
Covectors represent systems with $d-1$ fermions, i.e., systems with one hole.
Let $\{e_1,...,e_d\}$ be an orthonormal  basis in $V(d)$.
We use the notation
\begin{eqnarray}\label{111}
{\mathfrak E} _i=\star e_i=(-1)^{i-1}e_1\wedge ...\wedge e_{i-1}\wedge e_{i+1}...\wedge e_d.
\end{eqnarray}
It is easily seen that
\begin{eqnarray}\label{112}
e_i\wedge (\star e_j)=e_i\wedge {\mathfrak E} _j=\delta _{ij} {\cal E},
\end{eqnarray} 
where $\delta _{ij}$ is Kronecker's delta.
Physically, the $e_i\wedge {\mathfrak E} _i={\cal E}$ means that an electron fills the appropriate hole in the one-hole system represented by ${\mathfrak E} _i$,
and we get the system ${\cal E}$ which is full of electrons.
In $e_i\wedge {\mathfrak E} _j=0$ with $i\ne j$, an electron tries to fill a position in the one-hole system which is already occupied, and this is not possible.

Later we rewrite the right hand side of Eq.(\ref{com}), in terms of covectors (see Eq.(\ref{comp})).

\section{The join operation: joining states of holes}\label{meet}

The definition below is based on a split of a matrix $M$ into a submatrix $M_1$ that contains some of its columns, and another submatrix $M_2$ that contains the rest of its columns.

\begin{definition}
Given an extensor $A=a_1\wedge ...\wedge a_k$ in $\bigwedge {^k}[V(d)]$, we partition the set of indices into two subsets as
\begin{eqnarray}
\{1,...,k\}=\{i_1,...,i_h\}\cup \{j_1,...,j_{k-h}\};\;\;\;0\le h\le k.
\end{eqnarray}
A split of $A$ of class $(h, k-h)$ is its representation as
\begin{eqnarray}\label{55c}
A={\rm sgn}(A_1,A_2)  A_1\wedge A_2, 
\end{eqnarray}
where $A_1, A_2$ are extensors in $\bigwedge {^h}[V(d)]$, $\bigwedge {^{k-h}}[V(d)]$, correspondingly.
Depending on the order of the indices, ${\rm sgn}(A_1,A_2) $ is $1$ or $-1$.
For given $k,h$, the set of all such splits is denoted as  ${\cal S}_A(h,k-h)$.
\end{definition}

\begin{definition}\label{51}
Given two extensors $A=a_1\wedge ...\wedge a_k \in \bigwedge {^k}[V(d)]$ and
$B=b_1\wedge ...\wedge b_\ell \in \bigwedge {^\ell }[V(d)]$ the join operation $A\vee B$ (regressive product in Grassmann's terminology) is defined as:
\begin{itemize}
\item
if $k+\ell<d$ then $A\vee B=0$
\item
if $k+\ell\ge d$ then 
\begin{eqnarray}\label{56}
A\vee B=\sum _{{\cal S}_A(d-\ell,k+\ell -d)}{\rm sgn}(A_1,A_2)\det(A_1,B)A_2 \in \bigwedge {^{k+\ell -d}}[V(d)],
\end{eqnarray}
or equivalently
\begin{eqnarray}\label{56a}
A\vee B=\sum _{{\cal S}_B(\ell+k-d, d-k)}{\rm sgn}(B_1,B_2)\det(A,B_2)B_1 \in \bigwedge {^{k+\ell -d}}[V(d)].
\end{eqnarray}
In the first sum we take all splits in the set ${\cal S}_A(d-\ell,k+\ell -d)$, and in the second sum all splits in the set ${\cal S}_B(\ell+k-d, d-k)$.
\end{itemize}
For more general tensors than extensors, the join is defined through the multi-linearity property.
\end{definition}
The equivalence of the two expressions in Eqs.(\ref{56}), (\ref{56a}) is proved in \cite{R1}
(as we explained earlier, the notation $\vee, \wedge$ in \cite{R0,R1,R2,R3} is replaced here by $\wedge, \vee$, correspondingly).
The physical interpretation of these equations is discussed later.

A special case is when $A\in \bigwedge {^k}[V(d)]$ and $B\in {\bigwedge ^{d-k}}[V(d)]$. In this case 
\begin{eqnarray}\label{35}
A\vee B=\det (A,B).
\end{eqnarray}
Physically the $\det (A,B)$ represents the fermionic vacuum.
In particular
\begin{eqnarray}
e_1\vee ...\vee e_d=1.
\end{eqnarray}
If  $A\in \bigwedge {^k}[V(d)]$ then $\star A\in {\bigwedge ^{d-k}}[V(d)]$, and 
\begin{eqnarray}
A\vee (\star A)=\det (A,\star A).
\end{eqnarray}

\subsection{Properties of the join operation}

The following propositions are given without proof\cite{R1}, but we interpret them physically.

\begin{proposition}
de Morgan's theorem in the present context states that:
\begin{eqnarray}\label{59}
\star(A\vee B)=(\star A) \wedge (\star B);\;\;\;\star (A\wedge B)=(\star A)\vee (\star B).
\end{eqnarray}
\end{proposition}
Using this and the covectors defined in Eq.(\ref{111}), we rewrite Eq.(\ref{com}) for the Hodge star complement as
\begin{eqnarray}\label{comp}
\star \left [e_{i_1}\wedge ...\wedge e_{i_k}\right ]={\mathfrak E}_{i_1}\vee ...\vee {\mathfrak E}_{i_k}.
\end{eqnarray}
This shows clearly that the Hodge star complement replaces the fermions with holes, and the meet operation for fermions with the join operation for holes.

\paragraph*{Physical interpretation:}
\begin{itemize}
\item
In Eq.(\ref{55c}) the subsystem represented by $A$ which has $k$ fermions, is divided into two subsystems represented by $A_1$ and $A_2$, with $h$ and $k-h$ fermions, correspondingly.  

\item
In Eq.(\ref{56}), the subsystem $A_1$ has $d-\ell$ fermions (i.e., $\ell$ holes) and the subsystem $A_2$ has $k+\ell -d$ fermions (if $k+\ell -d\ge 0$). 
The $\ell$ holes of $A_1$ annihilate the $\ell$ fermions of $B$ and at the end the system has the  $k+\ell -d$ fermions of $A_2$. 
If $k+\ell -d<0$ then $A\vee B=0$. 

\item
In Eq.(\ref{56a}), the subsystem $B_2$ has $d-k$ fermions (i.e., $k$ holes) and the subsystem $B_1$ has $k+\ell -d$ fermions (if $k+\ell -d\ge 0$). 
The $k$ holes of $B_2$ annihilate the $k$ fermions of $A$ and at the end the system has the  $k+\ell -d$ fermions of $B_1$. 
If $k+\ell -d<0$ then $A\vee B=0$. 

\item
In Eq.(\ref{35}) the subsystem represented by $A$ has $k$ fermions, and the subsystem represented by $B$ has $d-k$ fermions, i.e., $k$ holes.
 The $k$ holes of $B$ annihilate the $k$ fermions of $A$ and the result is a scalar quantity, representing the fermionic vacuum.

\item
Eq.(\ref{59}) shows that the Hodge star complement converts the join (meet) operation of two fermionic subsystems, into the meet (join) operation between the corresponding subsystems of holes.
\end{itemize}

\begin{proposition}\label{pro2}
Let $A,B$ be extensors associated with the subspaces $V(k,A), V(\ell, B)$.
\begin{itemize}
\item
If  ${\rm span}[V(k,A)\cup V(\ell, B)]\ne V(d)$, then $A\vee B=0$.

\item
If  ${\rm span}[V(k,A)\cup V(\ell, B)]= V(d)$, and if $V(k,A)\cap V(\ell, B)\ne {\cal O}$, then $A\vee B\ne 0$  is the extensor associated with the subspace $V(k,A)\cap V(\ell, B)$.\end{itemize}
\end{proposition}
\paragraph*{Physical interpretation:}
The join operation finds the common part of two fermionic subsystems which 'live' in subspaces that span the whole space $V(d)$.
If the two fermionic subsystems 'live' in subspaces which do not span the whole space $V(d)$, then their join is zero. 

\begin{cor}
If $A$ is an extensor with ${\rm step}(A) \le d-1$ then  $A\vee A=0$. Also ${\cal E}\vee {\cal E}={\cal E}$.
\end{cor}
\paragraph*{Physical interpretation:}
We cannot join a system of holes with itself because this violates the Pauli exclusion principle.
In ${\cal E}\vee {\cal E}={\cal E}$ we join the vacuum of holes with itself.

From propositions \ref{pro1}, \ref{pro2} follows the following corollary.
\begin{cor}\label{45}
\begin{eqnarray}\label{55}
&&A\vee B\ne0\;\;\rightarrow\;\;A\wedge B= 0
\end{eqnarray}
\end{cor}
\paragraph*{Physical interpretation:}
If $A\vee B \ne 0$, these subsystems have a common part, and they cannot coexist in the same fermionic system. The
$A\wedge B=0$ expresses exactly this fact.
 Therefore corollary \ref{45} is related to the Pauli exclusion principle.

\begin{proposition}
\begin{itemize}

\item[(1)]
The associativity property for the meet, holds:
\begin{eqnarray}
(A\vee B)\vee C=A\vee (B\vee C).
\end{eqnarray}
\item[(2)]
The multilinearity property holds:
\begin{eqnarray}
(\lambda A_1+\mu A_2)\vee B=\lambda A_1\vee B+\mu A_2\vee B;\;\;\;\lambda, \mu \in {\mathbb C}.
\end{eqnarray}
\item[(3)]
If $A, B$ are extensors with ${\rm step}(A)=k$ and ${\rm step}(B)=\ell$ then
\begin{eqnarray}
A\vee B=(-1)^{(d-k)(d-\ell)} B\vee A.
\end{eqnarray}

\item[(4)]
\begin{eqnarray}\label{a}
e_1\wedge ...\wedge e_k={\mathfrak E} _{k+1}\vee ...\vee {\mathfrak E} _d.
\end{eqnarray} 
\item[(5)]
If $A, B, C$ are extensors with ${\rm step}(A)+{\rm step}(B)+{\rm step}(C)=d$ 
\begin{eqnarray}\label{28}
A\vee (B\wedge C)=A\wedge (B\vee C)=\det(A,B,C)
\end{eqnarray}
\end{itemize}
\end{proposition}

\paragraph*{Physical interpretation:}
\begin{itemize}
\item
In Eq.(\ref{a}) we have an extensor which is the meet of $e_1,...,e_k$ and represents a system of $k$ fermions.
The same system is also the join of the $d-k$ holes (absent fermions) ${\mathfrak E} _{k+1},...,{\mathfrak E} _d$.

\item
In Eq.(\ref{28}) we have three subsystems with a total number of fermions equal to $d$.
The $B\wedge C$ is a system with ${\rm step}(B)+{\rm step}(C)$ fermions.
The system $A$ has ${\rm step}(A)=d-[{\rm step}(B)+{\rm step}(C)]$ fermions, or equivalently ${\rm step}(B)+{\rm step}(C)$ holes.
The $A\vee (B\wedge C)$ annihilates the fermions in $B\wedge C$ with the holes in $A$, and we get the $\det(A,B,C)$.
Similar interpretation can be given to $A\wedge (B\vee C)$.
\end{itemize}

The following remark is useful for practical calculations of the join operation.
\begin{remark}
If $e_1,...,e_d$ is an orthonormal  basis in $V(d)$ we consider the extensors
\begin{eqnarray}
A=e_{i_1}\wedge...\wedge e_{i_r};\;\;\;B=e_{j_1}\wedge...\wedge e_{j_s}.
\end{eqnarray}
where the $\{i_1,...,i_r\}$ and $\{j_1,...,j_s\}$ are subsets of the set of all indices $\{1,...,d\}$.
Then
\begin{itemize}
\item
if $\{i_1,...,i_r\}\cup \{j_1,...,j_s\}\ne \{1,...,d\}$ then $A\vee B=0$ 
\item
if $\{i_1,...,i_r\}\cup \{j_1,...,j_s\}=\{1,...,d\}$ then $A\vee B$ is given by 
\begin{eqnarray}
A\vee B=\lambda \bigwedge _{k\in S}e_k;\;\;\;S=\{i_1,...,i_r\} \cap \{j_1,...,j_s\}
\end{eqnarray}
where $\lambda=\pm1$ depending on the order of the indices.
If $S$ is the empty set, the result is simply $\lambda$.
\end{itemize}
In particular, for every extensor $A$
\begin{eqnarray}
A\vee {\cal E}=A,
\end{eqnarray}
Physically ${\cal E}$ is the vacuum for holes and the join operation joins together the holes in $A$ with the vacuum for holes.
For $A=1$, we get 
\begin{eqnarray}
1\vee {\cal E}=1.
\end{eqnarray}
Furthermore, for every extensor $A$ with ${\rm step}(A) \le d-1$
\begin{eqnarray}
A\vee 1=0.
\end{eqnarray}
Physically, the state $1$ is full of holes and according to the Pauli exclusion principle, there is no room for extra holes.
The $A\vee 1=0$ expresses the fact that this is physically impossible.
These relations are summarized in table \ref{t1}.
\end{remark}
\begin{example}
We consider $\bigwedge [V(2)]$ and the basis $1, e_1, e_2, {\cal E}$. 
In table \ref{tq} we present the $A\vee B$ where $A,B$ are extensors in this basis.
Using this table and the multilinearity property, we can calculate the $A\vee B$ for all tensors in  $\bigwedge [V(2)]$.

$A\wedge B$ joins states of fermions, and $A\vee B$ joins states of holes.
Therefore $A\wedge B$ can be viewed as an OR gate for fermions, or equivalently as an AND gate for holes.
Similarly $A\vee B$ can be viewed as an AND gate for fermions, or equivalently as an OR gate for holes.
A comparison with Boolean (classical) gates is discussed in section \ref{gates} and table \ref{tt}.
\end{example}

\begin{example}
In $\bigwedge [V(3)]$ with orthonormal basis $e_1, e_2, e_3$, we consider the extensor $X$ in Eq.(\ref{79}) and the extensor $Z=e_1\wedge e_2$. We will calculate $X\vee Z$.

Using the definition \ref{51} we get 
\begin{eqnarray}
(e_1\wedge e_2)\vee (e_1\wedge e_2)=0;\;\;\;
(e_1\wedge e_3)\vee (e_1\wedge e_2)=-e_1;\;\;\;
(e_2\wedge e_3)\vee (e_1\wedge e_2)=-e_2.
\end{eqnarray}
Therefore
\begin{eqnarray}
X\vee Z&=&\alpha \gamma (e_1\wedge e_2)\vee (e_1\wedge e_2)+
\alpha \delta (e_1\wedge e_3)\vee (e_1\wedge e_2)+\beta \delta  (e_2\wedge e_3)\vee (e_1\wedge e_2)\nonumber\\&=&
-\alpha \delta e_1-\beta \delta  e_2.
\end{eqnarray}
\end{example}

\begin{example}
In $\bigwedge [V(4)]$ with orthonormal basis $e_1, e_2, e_3, e_4$, we calculate the following:
\begin{eqnarray}
(e_1\wedge e_2)\vee (e_3\vee e_4)=[(e_1\wedge e_2)\vee e_3]\vee e_4=0\vee e_4=0.
\end{eqnarray}
Here ${\rm step}(e_1\wedge e_2)+{\rm step}(e_3)=2+1<4$ and therefore according to the definition of the meet $(e_1\wedge e_2)\vee e_3=0$.
Also
\begin{eqnarray}
&&(e_1\wedge e_2)\vee (e_3\wedge e_4)=1\nonumber\\
&&(e_1\wedge e_2)\vee (e_3\wedge e_4 \wedge e_1)=e_1.
\end{eqnarray}
\end{example}

\section{Scalar product}\label{sca}

Let $A$ be an element in $\bigwedge [V(d)]$:
\begin{eqnarray}
A=\sum a(i_1,...,i_k)e_{i_1}\wedge ...\wedge e_{i_k}.
\end{eqnarray}
We use the notation
\begin{eqnarray}
{\overline A}=\sum {\overline a}(i_1,...,i_k)e_{i_1}\wedge ...\wedge e_{i_k}.
\end{eqnarray}
where ${\overline a}(i_1,...,i_k)$ is the complex conjugate of $a(i_1,...,i_k)$.

\begin{definition}
If $A,B \in \bigwedge {^k}[V(d)]$, their scalar product is 
\begin{eqnarray}
(A,B)=\overline A \vee (\star B)=\det(\overline A,\star B).
\end{eqnarray}
\end{definition}
\begin{proposition}
\mbox{}
\begin{itemize}
\item[(1)]
 If $A,B \in \bigwedge {^k}[V(d)]$ then $(A,B)=\overline {(B,A)}$.
\item[(2)]
As expected from vectors in an orthonormal basis
\begin{eqnarray}\label{36}
(e_i, e_j)=\delta _{ij}.
\end{eqnarray}
\item[(3)]
If
\begin{eqnarray}
A=\sum a_ie_i;\;\;\;B=\sum b_ie_i;\;\;\;A,B\in \bigwedge {^1}[V(d)],
\end{eqnarray}
then $(A,B)=\sum {\overline a}_i{b}_i$ is the standard scalar product of these vectors.
\end{itemize}
\end{proposition}
\begin{proof}
\begin{itemize}
\item[(1)]
We will prove that ${\overline A}\vee (\star {B})=\overline {{\overline B}\vee (\star {A})}={B}\vee (\star \overline A)$.
Both of them are scalar quantities. We get:
\begin{eqnarray}
\star [\overline A\vee (\star B)]=(\star \overline A)\wedge  B (-1)^{k(d-k)}=B\wedge (\star {\overline A})=[B\vee (\star {\overline A})]{\cal E}.
\end{eqnarray}
We take the Hodge star complement of both sides and since $B\vee (\star {\overline A})$ is a scalar quantity, we get
\begin{eqnarray}
\overline A\vee (\star B)=B\vee (\star \overline A).
\end{eqnarray}
\item[(2)]
Using Eq.(\ref{111}) we get
\begin{eqnarray}
(e_i,e_j)&=&e_i\vee (\star e_j)=(-1)^{i-1}e_i\vee (e_1\wedge ...\wedge e_{j-1}\wedge e_{j+1}\wedge ...\wedge e_d)\nonumber\\&=&(-1)^{i+j-2}\det (e_1\wedge ...\wedge e_{j-1}\wedge e_i\wedge e_{j+1}\wedge ...\wedge e_d)=
\delta _{ij}.
\end{eqnarray}
\item[(3)]
This follows immediately from Eq.(\ref{36}).
\end{itemize}
\end{proof}

\section{Fermionic gates versus Boolean gates}\label{gates}

In this section we compare the fermionic gates in table \ref{tq}, with Boolean gates.
From the outset, it is clear that there are important differences.
The operations in fermionic gates are antisymmetric, and in Boolean gates symmetric.
Also in fermionic gates we have the Pauli exclusion principle.
Nevertheless the comparison leads to interesting and surprising conclusions.

\subsection{Boolean gates}

Boolean logic is defined on the powerset $2^S$ of a set $S$ with $d$ elements (i.e., on the set of the $2^d$ subsets of $S$). 
Its three basic operations are:
\begin{itemize}
\item
Logical OR: the union $A\cup B$ ($A,B\subseteq S$).
\item
Logical AND: the intersection $A\cap B$.
\item
Logical NOT: the complement or negation $\neg A=S\setminus A$.
\end{itemize}
The powerset $2^S$ with these operations is a Boolean algebra (complemented distributive lattice). 
The least element is the empty set $\emptyset$, and the greatest element is $S$.
The partial order $\prec$ in this lattice is the subset $\subseteq$.

A Boolean gate with two inputs and one output,  is a function  (e.g., \cite{vou}) 
\begin{eqnarray}
{\mathfrak G}:\;\;2^S \times 2^S\;\;\rightarrow\;\;2^S
\end{eqnarray}
that maps an input $(A_1,A_2)$ to an output.
Examples are the OR, AND gates:
\begin{eqnarray}
{\mathfrak G}_{\rm OR}(A_1, A_2)=A_1\cup A_2;\;\;\;{\mathfrak G}_{\rm AND}(A_1, A_2)=A_1\cap A_2;\;\;\;A_1,A_2\in 2^S.
\end{eqnarray}
The NOT gate has one input and one output:
\begin{eqnarray}
{\mathfrak G}_{\rm NOT}(A_1)=S\setminus A_1.
\end{eqnarray}

For comparison with the fermionic gates in table \ref{tq}, we consider Boolean algebra based on the powerset:
\begin{eqnarray}
2^S=\{\emptyset, \{1\}, \{2\}, S\};\;\;\;S=\{1,2\}.
\end{eqnarray}
In table \ref{tt} we present Boolean algebra results for ${\mathfrak G}_{\rm OR}(A_1, A_2)=A_1\cup A_2$ and ${\mathfrak G}_{\rm AND}(A_1, A_2)=A_1\cap A_2$, where $A_1, A_2\in 2^S$.

\subsection{Fermionic gates}
We consider $\bigwedge [V(2)]$ and its basis 
\begin{eqnarray}
{\cal B}=\{1, e_1, e_2, {\cal E}\}.
\end{eqnarray}
There is a bijective map ${\cal M}$ from the powerset $2^S$  to the basis ${\cal B}$, as follows:
\begin{eqnarray}\label{aa}
{\cal M}(\emptyset)=1 ;\;\;\;\; {\cal M}(\{1\})=e_1;\;\;\;\; {\cal M}(\{2\})=e_2;\;\;\;\;{\cal M} (S)={\cal E}=e_1\wedge e_2.
\end{eqnarray}

Sometimes in mathematics, we define partial binary operations on a set $\Omega$, i.e., maps from a subset ${\cal D}$ of $\Omega \times \Omega$ to $\Omega$.
The operation is not defined for all pairs of elements. An example is groupoids.

We introduce two partially defined `pseudo-fermionic operations' on $2^S$, 
which we denote as $\curlywedge$ and $\curlyvee$, as follows.
Let 
\begin{eqnarray}
&&{\cal D}_1=\{(\emptyset, \emptyset ), (\{1\}, \emptyset ), (\{2\}, \emptyset ), (S, \emptyset ), (\emptyset, \{1\}), (\emptyset, \{2\}), (\{1\},\{2\}), (\emptyset, S)\} \subset 2^S\times 2^S\nonumber\\
&&{\cal D}_2=\{(S, \emptyset ), (S, \{1\}), (\{1\}, \{2\} ), (S, \{2\}), (\emptyset, S), (\{1\}, S), (\{2\}, S), (S, S)\} \subset 2^S\times 2^S.
\end{eqnarray}
$A_1 \curlywedge A_2$ is the following map from ${\cal D}_1$ to $2^S$:
\begin{eqnarray}
&&A_1 \curlywedge A_2={\cal M}^{-1}[{\cal M}(A_1)\wedge {\cal M}(A_2)]\nonumber\\
&&(A_1, A_2)\in {\cal D}_1;\;\;\;{\cal M}(A_1), {\cal M}(A_2)\in {\cal B}.
\end{eqnarray} 
The result ${\cal M}(A_1)\wedge {\cal M}(A_2)$ might be the zero vector in $\bigwedge [V(2)]$, 
or it might be one of the vectors in the basis  ${\cal B}$ with a minus sign.
They belong to $\bigwedge [V(2)]$ but they do not belong to the basis ${\cal B}$, and the  ${\cal M}^{-1}[{\cal M}(A_1)\wedge {\cal M}(A_2)]$ is not defined.
For this reason we have excluded such pairs from ${\cal D}_1$.

$A_1 \curlyvee A_2$ is the following map from ${\cal D}_2$ to $2^S$:
\begin{eqnarray}
&&A_1 \curlyvee A_2={\cal M}^{-1}[{\cal M}(A_1)\vee {\cal M}(A_2)]\nonumber\\
&&(A_1, A_2)\in {\cal D}_2;\;\;\;{\cal M}(A_1), {\cal M}(A_2)\in {\cal B}.
\end{eqnarray} 

Using table \ref{tq} for the exterior calculations ${\cal M}(A_1)\wedge {\cal M}(A_2)$ and ${\cal M}(A_1)\vee {\cal M}(A_2)$, we calculated 
$A_1 \curlywedge A_2$ and $A_1 \curlyvee A_2$. The results are shown in table \ref{tt}.  
It is seen that in the domain ${\cal D}_1$ the $A_1 \curlywedge A_2$ is the same as  $A_1\cup A_2$ (it is very different from $A_1\cap A_2$)
Also in domain ${\cal D}_2$ the $A_1 \curlyvee A_2$ is the same as  $A_1\cap A_2$ (it is very different from $A_1\cup A_2$).

We have already emphasized that the operations in Boolean algebra are symmetric, whilst in exterior calculus are antisymmetric.
However, it is seen that $A_1\curlywedge A_2$ ($A_1\curlyvee A_2$) is a type of Boolean OR (AND) gate, restricted in the domain ${\cal D}_1$ (${\cal D}_2$).
Experimental implementation of these mathematical operations with novel electronic devices is highly desirable, because they can lead to simulation 
of exterior calculus (fermionic logic) with electronic devices.

\begin{table}
\caption{$A_1\curlywedge A_2$ in the domain ${\cal D}_1$ and $A_1\curlyvee A_2$ in the domain ${\cal D}_2$.
$A_1\curlywedge A_2$ is the same as $A_1\cup A_2$ (in the domain ${\cal D}_1$), and 
$A_1\curlyvee A_2$ is the same as $A_1\cap A_2$ (in the domain ${\cal D}_2$).
Here $A_1,A_2 \subseteq S$ with $S=\{1,2\}$.}
\def\arraystretch{2}
\begin{tabular}
{|c|c|c|c|c|c|}\hline
\multicolumn{2}{|c|}{$2^S\times 2^S$}&\multicolumn{2}{|c|}{Boolean algebra}&\multicolumn{2}{c|}{Exterior calculus}\\\hline
$\;\;A_1\;\;$&$\;\;A_2\;\;$&$A_1\cup A_2$&$A_1\cap A_2$&$A_1\curlywedge A_2$ in ${\cal D}_1$&$A_1\curlyvee A_2$ in ${\cal D}_2$\\\hline
$\emptyset $&$\emptyset$&$\emptyset$&$\emptyset$&$\emptyset$&$$\\\hline
$\{1\}$&$\emptyset$&$\{1\}$&$\emptyset$&$\{1\}$& $$\\\hline
$\{2\}$&$\emptyset$&$\{2\}$&$\emptyset$&$\{2\}$&$$\\\hline
$S$&$\emptyset$&$S$&$\emptyset$&$S$&$\emptyset$\\\hline
$\emptyset $&$\{1\} $&$\{1\}$&$\emptyset$& $\{1\}$&$$\\\hline
$\{1\}$&$\{1\}$&$\{1\}$&$\{1\}$& $$&$$\\\hline
$\{2\}$&$\{1\} $&$S$&$\emptyset$& $$&$$\\\hline
$S$&$\{1\}$&$S$&$\{1\}$& $$& $\{1\}$\\\hline
$\emptyset $&$\{2\} $&$\{2\}$&$\emptyset$& $\{2\}$&$$\\\hline
$\{1\}$&$\{2\}$&$S$&$\emptyset$&$S$& $\emptyset $\\\hline
$\{2\}$&$\{2\}$&$\{2\}$&$\{2\}$& $$&$$\\\hline
$S$&$\{2\}$&$S$&$\{2\}$& $$& $\{2\}$\\\hline
$\emptyset $&$S$&$S$&$\emptyset$& $S$& $\emptyset$\\\hline
$\{1\}$&$S$&$S$&$\{1\}$& $$& $\{1\}$\\\hline
$\{2\}$&$S$&$S$&$\{2\}$& $$& $\{2\}$\\\hline
$S$&$S$&$S$&$S$&$$&$S$\\\hline
\end{tabular}\label{tt}
\end{table}

\section{Simulation of fermionic logic within multi-qubit systems}

Let $h_2$ be a two-dimensional Hilbert space describing a qubit.
A $d$-qubit system is described by the $2^d$-dimensional Hilbert space ${\mathfrak H}_d=h_2\otimes ...\otimes h_2$.
There is a well known map between fermionic and multi-qubit systems, and the Jordan-Wigner transformation expresses fermionic operators in terms of Pauli matrices acting on multi-qubits.

Here we define fermionic logic (exterior calculus)  within a  multi-qubit system.
Fermionic logic is important for fast computation of fermionic calculations.
Its implementation with genuine fermionic gates might be a longer term task experimentally, while its implementation with
multi-qubit systems (discussed below) might be a more realistic short term task.

\subsection{A bijective map between fermionic and multi-qubit systems}

We consider a subset $\{i_1,...i_k\}$ of the set $\{1,...,d\}$. 
The $i_j$ are taken in ascending order, so that $i_1<...<i_k$.
$\ket{{\cal S}(i_1,...,i_k)}$ denotes a $d$-qubit system, with the $k$ qubits labelled $i_1,...,i_k$ in the state $\ket{1}$, and the rest $d-k$ qubits in the state $\ket{0}$.
For example, if $d=4$ then $\ket{{\cal S}(1,3)}=\ket{1,0,1,0}$.

The $2^d$ kets $\ket{{\cal S}(i_1,...,i_k)}$ form a basis ${\mathfrak B}_d$ in ${\mathfrak H}_d$, and the  $2^d$ extensors $e_{i_1}\wedge...\wedge e_{i_k}$ form a basis
${\cal B}_d$ in $\bigwedge [V(d)]$.
There is a bijective map ${\cal N}_d$ between the basis ${\mathfrak B}_d$ and the basis ${\cal B}_d$ as follows:
\begin{eqnarray}\label{123}
{\cal N}_d [\ket{{\cal S}(i_1,...,i_k)} ]=e_{i_1}\wedge...\wedge e_{i_k}.
\end{eqnarray}
Special cases of this map are
\begin{eqnarray}\label{bb}
{\cal N}_d(\ket{0,...,0})=1;\;\;\;{\cal N}_d(\ket{1,...,1})={\cal E}.
\end{eqnarray}
This defines a bijective map ${\cal N}_d$ between ${\mathfrak H}_d$ and $\bigwedge [V(d)]$, where the superpositions in 
 ${\mathfrak H}_d$ are mapped into the corresponding superpositions in $\bigwedge [V(d)]$.
We note that there is a zero vector in both ${\mathfrak H}_2$ and  $\bigwedge [V(2)]$ (which is not the vacuum), and that ${\cal N}_2(0)=0$.

The notation for an arbitrary $d$ is clearly complex, and there is a danger that physics will be hidden in the complex mathematical notation. 
For this reason we consider below the example with $d=2$.
In this case we have the four-dimensional space ${\mathfrak H}_2$ for the two-qubit system, and the four-dimensional space $\bigwedge [V(2)]$ for the two-fermion system.
In ${\mathfrak H}_2$ we consider the orthonormal basis 
\begin{eqnarray}
{\mathfrak B}_2=\{\ket{0,0}, \ket{1,0}, \ket{0,1}, \ket{1,1}\},
\end{eqnarray}
and in $\bigwedge [V(2)]$ the orthonormal basis 
\begin{eqnarray}
{\cal B}_2=\{1, e_1, e_2, {\cal E}\}.
\end{eqnarray}
There is a bijective map ${\cal N}_2$ between the basis ${\mathfrak B}_2$ of ${\mathfrak H}_2$ and the basis ${\cal B}_2$ of $\bigwedge [V(2)]$ as follows:
\begin{eqnarray}\label{bb}
{\cal N}_2(\ket{0,0})=1 ;\;\;\;\; {\cal N}_2(\ket{1,0})=e_1;\;\;\;\; {\cal N}_2(\ket{0,1})=e_2;\;\;\;\;{\cal N} _2(\ket{1,1})={\cal E}=e_1\wedge e_2.
\end{eqnarray}
This defines a bijective map ${\cal N}_2$ between ${\mathfrak H}_2$ and $\bigwedge [V(2)]$,
where the superpositions in 
 ${\mathfrak H}_2$ are mapped into the corresponding superpositions in $\bigwedge [V(2)]$:
\begin{eqnarray}\label{bbb}
{\cal N}_2[a\ket{0,0}+b\ket{1,0}+c\ket{0,1}+d{\cal E}]=a+be_1+ce_2+d{\cal E};\;\;\;a,b,c,d \in {\mathbb C}.
\end{eqnarray}
In particular, the zero vector in  ${\mathfrak H}_2$ is mapped into the zero vector in $\bigwedge [V(2)]$.
This map does not preserve the scalar product. The  scalar product in ${\mathfrak H}_2$, does not correspond to the scalar product in $\bigwedge [V(2)]$ (which is defined in section \ref{sca} within exterior calculus).
For example, no scalar product is defined between $e_1$ and ${\cal E}$, because the first belongs in $\bigwedge ^{1}[V(2)]$ and the second in $\bigwedge ^{2}[V(2)]$. 
But the scalar product between the corresponding states $\ket{1,0}$ and $\ket{1,1}$ in  ${\mathfrak H}_2$, is well defined.

Isomorphism between two spaces requires not only a bijective map, but also some basic operations to be preserved.
It is important to specify explicitly what these operations are. 
The ${\mathfrak H}_2$ and  $\bigwedge [V(2)]$ are not isomorphic, with respect to the scalar product.

\subsection{Fermionic formalism in a  multi-qubit system}\label{fermi}
The bijective map ${\cal N}_2$ in Eqs.(\ref{bb}),(\ref{bbb}) can be used to define exterior calculus within the multi-qubit space ${\mathfrak H}_2$.
\begin{definition}
The `fermionic join', `fermionic meet' and `star hodge complement', of multi-qubit states $\ket{s_1}, \ket{s_2} \in {\mathfrak H}_2$, is given by:
\begin{eqnarray}
&&\ket{s_1}\vee \ket{s_2}={\cal N}_2^{-1}[{\cal N}_2(\ket{s_1})]\vee {\cal N}_2^{-1}[{\cal N}_2(\ket{s_2})]={\cal N}_2^{-1}[{\cal N}_2(\ket{s_1})\vee {\cal N}_2(\ket{s_2})]\nonumber\\
&&\ket{s_1}\wedge \ket{s_2}={\cal N}_2^{-1}[{\cal N}_2(\ket{s_1})]\wedge {\cal N}_2^{-1}[{\cal N}_2(\ket{s_2})]={\cal N}_2^{-1}[{\cal N}_2(\ket{s_1})\wedge {\cal N}_2(\ket{s_2})]\nonumber\\
&&\star \ket{s}=\star {\cal N}_2^{-1}[{\cal N}_2(\ket{s})]={\cal N}_2^{-1}[\star {\cal N}_2(\ket{s})]
\end{eqnarray}
For simplicity we use the same notation for fermionic join (fermionic meet) in both  $\bigwedge [V(2)]$ and ${\mathfrak H}_2$.
If either ${\cal N}_2(\ket{s_1})\wedge {\cal N}_2(\ket{s_2})=0$ or ${\cal N}_2(\ket{s_1})\vee {\cal N}_2(\ket{s_2})=0$ (the zero vector in $\bigwedge [V(2)]$)
then ${\cal N}_2^{-1}(0)=0$ (the zero vector in ${\mathfrak H}_2$)
and we say that these operations are physically impossible in the multi-qubit systems.
This introduces a 'Pauli-like principle' into the multi-qubit system. 
\end{definition}

In table \ref{tg} we present results for the fermionic join and fermionic meet of states in ${\mathfrak H}_2$.
 The result $0$ is the zero vector, and it indicates that this operation is physically impossible. 
 The antisymmetry property is seen in this table. For example, $\ket{0,1} \wedge \ket{1,0}=-\ket{1,0}\wedge \ket{0,1}$ and $\ket{0,1} \vee \ket{1,0}=-\ket{1,0}\vee \ket{0,1}$.
 The 'Pauli-like principle'  is also seen in the table. For example,
 $\ket{s}\wedge \ket{s}=0$ when $\ket{s}\ne \ket{0,0}$, and $\ket{s}\vee \ket{s}=0$ when $\ket{s}\ne \ket{1,1}$.
 
 Implementation of fermionic logic (exterior calculus) in the context of a multi-qubit system, might be a more realistic experimental task, than its implementation within a genuine fermionic system.
  
\begin{table}
\caption{Fermionic join $\ket{s_1} \wedge \ket{s_2}$ and fermionic meet $\ket{s_1} \vee \ket{s_2}$ of multi-qubit states in ${\mathfrak H}_2$. 
The result $0$ indicates that this operation is physically impossible, and introduces a `Pauli-like' exclusion principle into the system of multi-qubits.}
\def\arraystretch{2}
\begin{tabular}
{|c|c|c|c|}\hline
$\;\;\ket{s_1}\;\;$&$\;\;\ket{s_2}\;\;$& $\ket{s_1} \wedge \ket{s_2}$&$\ket{s_1} \vee \ket{s_2}$\\\hline
$\ket{0,0}$&$\ket{0,0}$&$\ket{0,0}$&$0$\\\hline
$\ket{0,0}$&$\ket{1,0}$&$\ket{1,0}$&$0$\\\hline
$\ket{0,0}$&$\ket{0,1}$&$\ket{0,1}$&$0$\\\hline
$\ket{0,0}$&$\ket{1,1}$&$\ket{1,1}$&$\ket{0,0}$\\\hline
$\ket{1,0}$&$\ket{0,0}$&$\ket{1,0}$&$0$\\\hline
$\ket{1,0}$&$\ket{1,0}$&$0$&$0$\\\hline
$\ket{1,0}$&$\ket{0,1}$&$\ket{1,1}$&$\ket{0,0}$\\\hline
$\ket{1,0}$&$\ket{1,1}$&0&$\ket{1,0}$\\\hline
$\ket{0,1}$&$\ket{0,0}$&$\ket{0,1}$&$0$\\\hline
$\ket{0,1}$&$\ket{1,0}$&$-\ket{1,1}$&$-\ket{0,0}$\\\hline
$\ket{0,1}$&$\ket{0,1}$&$0$&$0$\\\hline
$\ket{0,1}$&$\ket{1,1}$&$0$&$\ket{0,1}$\\\hline
$\ket{1,1}$&$\ket{0,0}$&$\ket{1,1}$&$\ket{0,0}$\\\hline
$\ket{1,1}$&$\ket{1,0}$&$0$&$\ket{1,0}$\\\hline
$\ket{1,1}$&$\ket{0,1}$&$0$&$\ket{0,1}$\\\hline
$\ket{1,1}$&$\ket{1,1}$&$0$&$\ket{1,1}$\\\hline
\end{tabular}\label{tg}
\end{table}

\section{Discussion}

We used exterior calculus with all its three operations joint, meet and Hodge star complement, for the study of fermion-hole systems.
We interpreted physically its mathematical theorems, and used it as a formalism for fermionic quantum logic.
A summary of the basic relations has been presented in table \ref{t1}.

The formalism leads to fermionic versions of AND, OR gates, which are shown in table \ref{tq}.
A comparison of fermionic gates with Boolean gates in section \ref{gates} and table \ref{tt} can lead to novel electronic devices for the simulation of fermionic logic.

Using the map in Eq.(\ref{123}) between fermionic and multi-qubit systems, we have 
introduced fermionic (exterior calculus) logic in multi-qubit systems (see table \ref{tg}).
This leads to a simulation of fermionic gates within multi-qubit systems.

The work studies exterior calculus for the description of fermionic systems.
It then uses it, in two schemes for the simulation of fermionic gates.
The first scheme uses novel variants of Boolean gates, and the second one uses multi-qubits.
These schemes might be experimentally easier to implement, than fermionic gates with genuine fermionic systems.

Determinants and the Grassmann formalism are intimately related to antisymmetry which in Physics correspond to fermions.
On the other hand Physics introduced parafermions with intermediate statistics between bosons and fermions. 
It is a challenge to introduce 'para-determinants' and build a 'para-Grassmann formalism'.
 
\end{document}